\documentclass[letterpaper, 10pt, conference]{ieeeconf}      

\IEEEoverridecommandlockouts 
\usepackage[
style=ieee,       
backend=bibtex,   
isbn=false,       
doi=false,        
giveninits=true,  
url=false,
date=year,        
maxbibnames=99,   
minnames=1,       
maxnames=4,      
sorting=none,     
alldates=year,
]{biblatex}           
\AtEveryBibitem{\clearfield{pages}} 
\AtEveryBibitem{\clearfield{volume}} 
\AtEveryBibitem{\clearfield{number}} 

\usepackage{glossaries}
\usepackage{booktabs}

\usepackage{siunitx}				
\usepackage{amsmath,amssymb,amsfonts}
\usepackage{xfrac}

\usepackage{color}								
\usepackage[dvipsnames]{xcolor}

\definecolor{m1}{rgb}{0, 0.4470, 0.7410}
\definecolor{m2}{rgb}{0.8500, 0.3250, 0.0980}
\definecolor{gfl}{rgb}{0.6350 0.0780 0.1840}
\definecolor{gfr}{HTML}{00A99A}
\definecolor{lines}{rgb}{0.75, 0, 0.75}
\definecolor{econ}{rgb}{0.4660, 0.6740, 0.1880}

\usepackage{graphicx} 
\usepackage{tikz}
\usetikzlibrary{calc}
\usetikzlibrary{arrows}

\usepackage[german,USenglish,UKenglish]{babel}

\usepackage{enumerate}				
\usepackage[hidelinks]{hyperref}    

\usepackage[ruled, vlined, linesnumbered,english]{algorithm2e}
\SetKwRepeat{Do}{do}{while}%
\SetAlFnt{\small} 
\SetKwComment{Comment}{$\triangleright$\ }{}
\SetVlineSkip{2pt} 
\DontPrintSemicolon
\SetNlSty{}{}{:} 
\SetAlgoNlRelativeSize{-1} 
\SetCommentSty{} 

\newcommand{\R}{\mathbb{R}}
\newcommand{\N}{\mathbb{R}}
\newcommand{\col}[1]{\text{col}\{#1\}}
\newcommand{\diag}[1]{\text{diag}\{#1\}}
\newcommand{\mat}[1]{\begin{bmatrix}#1\end{bmatrix}}

\newcommand{\inci}{M}

\newcommand{\ext}{\text{elec}}
\newcommand{\econ}{\text{econ}}
\newcommand{\reff}{\text{ref}}
\newcommand{\rl}[1]{\bar{#1}}
\newcommand{\err}[1]{\tilde{#1}}

\newcommand{\numOfnodes}{n}
\newcommand{\numOflines}{m}
\newcommand{\numOfDGU}{d}
\newcommand{\numOfstates}{n_{\text{cl}}}
\newcommand{\setOfnodes}{\mathcal{B}}
\newcommand{\setOfedges}{\mathcal{E}}
\newcommand{\numOfMG}{n_{\text{mg}}}
\newcommand{\Pload}{p_{\text{L}}}
\newcommand{\iL}[1]{\ifthenelse{\equal{#1}{}}{i_{{\text{L}}}}{i_{{\text{L},#1}}}}

\newcommand{\V}{\mathcal{V}}
\newcommand{\I}{\mathcal{I}}

\newcommand{\X}{\mathcal{X}}
\newcommand{\Pl}{\mathcal{P}}
\newcommand{\M}{\mathcal{M}}

\newcommand{\Lpi}[1]{\ifthenelse{\equal{#1}{}}{L_\pi}{L_{{\pi,#1}}}}
\newcommand{\Rpi}[1]{\ifthenelse{\equal{#1}{}}{R_\pi}{R_{{\pi,#1}}}}
\newcommand{\Cpi}[1]{\ifthenelse{\equal{#1}{}}{C_\pi}{C_{{\pi,#1}}}}

\newcommand{\Lf}[1]{\ifthenelse{\equal{#1}{}}{L_\text{f}}{L_{{\text{f},#1}}}}
\newcommand{\Rf}[1]{\ifthenelse{\equal{#1}{}}{R_\text{f}}{R_{{\text{f},#1}}}}
\newcommand{\Cf}[1]{\ifthenelse{\equal{#1}{}}{C_\text{f}}{C_{{\text{f},#1}}}}
\newcommand{\ir}[1]{\ifthenelse{\equal{#1}{}}{i_\text{f}}{i_{{\text{f},#1}}}}

\newcommand{\glob}{\text{ext}}
\newcommand{\loc}{\text{loc}}

\newacronym{res}{RES}{renewable energy source}
\newacronym{sos}{SOS}{Sum-of-squares}
\newacronym{dgu}{DGU}{distributed generation unit}
\newacronym{lmi}{LMI}{linear matrix inequality}

\newtheorem{lemma}{Proposition}
\newtheorem{remark}{Remark}

\newtheorem{definition}{Definition}
\newtheorem{theorem}{Theorem}

\title{\LARGE \bf
	Passivity-based economic ports for optimal operation of networked DC microgrids 
}

\author{Pol Jané-Soneira, Albertus J. Malan, Ionela Prodan and Sören Hohmann
	\thanks{P. Jané-Soneira, A. J. Malan and S. Hohmann are with the Institute of Control Systems, Karlsruhe Institute of Technology, 76131, Karlsruhe, Germany. Corresponding author is Pol Jané-Soneira, {\tt \scriptsize pol.soneira@kit.edu}.}%
	\thanks{I. Prodan is with the Univ. Grenoble Alpes, Grenoble INP, LCIS, F-26000, Valence, France. I. Prodan’s research benefited from the support of the FMJH Program PGMO and from the support to this program from EDF.}%
}

\begin{document}

\maketitle
\thispagestyle{empty}
\pagestyle{empty}

\begin{abstract}
    In this paper, we introduce the novel concept of economic ports, allowing modular and distributed optimal operation of networked microgrids. Firstly, we design a novel price-based controller for optimal operation of a single microgrid and show asymptotic stability. Secondly, we define novel physical and economic interconnection ports for the microgrid and study the dissipativity properties of these ports. Lastly, we propose an interconnection scheme for microgrids via the economic ports. This interconnection scheme requires only an exchange of the local prices and allows a globally economic optimal operation of networked microgrids at steady state, while guaranteeing asymptotic stability of the networked microgrids via the passivity properties of economic ports. The methods are demonstrated through various academic examples. 
\end{abstract}

\section{Introduction} \label{sec:intro}

Power systems are undergoing a transformation towards a sustainable and emission-free electrical supply based on renewable energies. Due to the integration of volatile renewable energy sources and the removal of large-scale fossil fuel generators, power systems are experiencing a reduction in grid inertia and are hence increasingly facing grid stability issues. With the ongoing removal of large-scale generators, the \glspl{dgu} need to contribute to grid stability. Due to the large number of small-scale \glspl{dgu} in an energy system, both the optimal coordination of the \glspl{dgu} and the plug-and-play operation ensuring stability are crucial.  

In literature, many approaches propose a passivity-based controller design for \glspl{dgu} \cite{nahata2020passivity,strehle2020,watson2019stability}. These regulators achieve an offset-free regulation of a given voltage reference and have desirable plug-and-play properties while guaranteeing asymptotic stability of the overall interconnected system via passivity. Recently, extensions have been proposed in order to achieve current- \cite{nahata2022current} or power-sharing \cite{malan2023passivity} within the passivity-based framework, or approximate power-sharing considering a simultaneous voltage- and frequency control in AC systems \cite{ojo2023distributed}. Although allowing plug-and-play operation and ensuring asymptotic stability, passivity-based methods are in general purely decentralized approaches which cannot achieve an economically optimal operation or steer the system to an economically optimal steady state.  

Addressing this issue, \cite{stegink2016unifying, kolsch2020distributed} propose distributed passivity- and optimization-based controllers for a microgrid in port-Hamiltonian form that is able to steer the system to an economically optimal steady state. The intrinsic, favorable passivity properties of the port-Hamiltonian system enables plug-and-play operation while ensuring asymptotic stability. However, in both approaches, the whole microgrid is modeled as a synchronous generator, which is interconnected with other microgrids via lossless, static lines. These simplifications and assumptions, although allowing important theoretical contributions, hamper the application to low inertia microgrids with lossy lines, which will adopt a crucial role in future power systems. In \cite{jane2023mpc,soneira2022optimal}, a model predictive controller exploiting the passivity properties of the underlying controllers for ensuring modular stability is proposed. This method achieves an optimal operation, but requires considerable amount of computing power each time step and does not allow a distributed operation. In \cite{dorfler2015breaking,zhao2015distributed}, optimization-based controllers for AC and DC microgrids with droop-controllers are proposed. Although these methods are not based on passivity, asymptotic stability of an economically optimal steady state together with plug-and-play capabilities are shown. However, this again comes at the cost of considering a system model with limiting assumptions and approximations, e.g. static lines and single capacitance dynamics as DC microgrid model or a simple oscillator as AC microgrid node dynamics. In particular, dynamics of the \glspl{dgu}, transmission lines or nonlinear loads are not considered. 

\textit{Contributions:} We propose an optimization-based controller for DC microgrids that steers a microgrid to an economically optimal steady state in a distributed manner without knowledge of the loads or transmission line parameters. Conditions for asymptotic stability are provided. We further leverage these results to study the interconnection of various microgrids by introducing the novel concept of passivity-based economic ports. These economic ports have a cyber-physical nature, which differs from the typical physical interconnection ports defined with physical variables like voltages and currents in passivity-based control. The economic port allows interconnecting different microgrids on an information level to achieve overall economic optimality while ensuring plug-and-play stability in a distributed manner. 

The remainder of this paper is structured as follows. In Section~\ref{sec:system}, the system model for a general, converter-based microgrid considered in this work is presented. The distributed optimization-based controller design is presented in Section~\ref{sec:controller}. In Section~\ref{sec:ports}, we introduce the electric and economic ports allowing the microgrids to interconnect on a physical- and information-basis. In Section~\ref{sec:simulation}, the performance of the proposed controller for a cluster of networked microgrids is illustrated via simulations. 

\textit{Notation:} %
Lowercase letters $x \in \R^{n}$ represent vectors, and uppercase letters $X \in \R^{n\times n}$ represent matrices. The transpose of a vector $x \in \R^n$ is written as $x^\top$. The vector $x = \col{x_i}$ and matrix $X = [x] = \diag{x_i}$ are the $n\times 1$ column vector and $n \times n$ diagonal matrix of the elements $x_i$, $i = 1,\dots, n$, respectively. Let $I_{n}$ denote the $n \times n$ identity matrix and $1_n \in \R^n$ a vector of ones. Calligraphic letters $\X$ represent sets, and $\X \times \I$ denotes the cartesian product of the two sets. For vectors $x_{\min}$, $x_{\max} \in \R^n$, the set $\X = [x_{\min}, x_{\max}]$ is a shorthand notation for the convex polytope $\X = \{ x \in \R^n \, \vert \, x_{\min} \leq x \leq x_{\max} \}$, where $\leq$ holds component-wise. A directed graph is denoted by $\mathcal{G}(\setOfnodes,\mathcal{E})$, where $\setOfnodes$ is the set of nodes and $\mathcal{E} \subseteq \setOfnodes\times \setOfnodes$ the set of edges. The cardinality for a set $\setOfnodes$ is denoted by $\vert \setOfnodes \vert$. The incidence matrix $\inci \in \R^{\vert \setOfnodes \vert \times \vert \mathcal{E} \vert }$ is defined as $M = (m_{ji})$ with $m_{ji} = -1$ if edge $e_j \in \mathcal{E}$ leaves node $v_i \in \setOfnodes$, $m_{ji} = 1$ if edge $e_j \in \mathcal{E}$ enters node $v_i \in \setOfnodes$, and $m_{ji} = 0$ otherwise.

\section{System model} \label{sec:system}

In this paper, we consider a set of microgrids $k \in \M = \{ 1,\dots,\numOfMG\}$, each comprising a set $\setOfnodes^k$ of $\numOfnodes^k = \vert \setOfnodes^k \vert$ electrical buses or nodes connected via a set $\mathcal{E}^k$ of $\numOflines^k = \vert \mathcal{E}^k \vert$ electrical lines. Defining an arbitrary line current direction over the microgrid power lines, we describe the network topology of each microgrid with the directed graph $\mathcal{G}^k(\setOfnodes^k,\mathcal{E}^k)$, where $\setOfnodes^k$ is the set of nodes and $\mathcal{E}^k$ of edges. 
In the following, we present the dynamic models of the DC microgrid components. We consider nodes $i \in \setOfnodes_{\text{L},k} \subseteq \setOfnodes^k$ having only a nonlinear load, and nodes $i \in \setOfnodes_{\text{DGU},k} \subseteq \setOfnodes^k$ having additionally a \gls{dgu}, with $\setOfnodes_{\text{L},k} \cup \setOfnodes_{\text{DGU},k} = \setOfnodes^k$. When equipped with a \gls{dgu}, the node voltage $v_i$ can be directly influenced. 

In the remainder of this section, the dynamic models of the microgrid components are presented. The microgrid index $k$ (always displayed as superscript) is omitted for simplicity until further notice, since the same microgrid structure holds for all $k \in \M$ (microgrids may have different sizes, topologies and parameters). 

\subsection{Distributed generation unit (DGU) node} \label{sec:dgu}

A schematic representation of a node composed of \gls{dgu} and nonlinear load is shown in Figure~\ref{fig:dgu}. It consists of a Buck converter supplying a voltage $v_{\text{t},i} \in \R_{\geq 0}$, a nonlinear load with current $\iL{i}(v_i) \in \R$, and a filter with inductance $\Lf{i} \in \R_{>0}$, capacitance $\Cf{i} \in \R_{> 0}$, and resistance $ R_{\text{f},i} \in \R_{> 0}$. The Buck converter and the filter form a \gls{dgu}. Each bus $i \in \setOfnodes$ has two states, the node voltage $v_i$ and the filter current $i_{\text{f},i}$. The Buck converter voltage $v_{\text{t},i}$ is defined as the system input. Here, we use the well known averaging model for the Buck converter, by which we disregard the switching behavior~\cite{middlebrook1977}. The dynamics for every bus $i \in \setOfnodes_{\text{DGU}}$ are~\cite{nahata2020passivity,tucci2016}
\begin{subequations} \label{eq:dgu}
    \begin{align}
        \Cf{i} \dot{v}_{i} & = \ir{i} - \iL{i}(v_i) - i_{\text{ext},i} \label{eq:voltage_dynamics_dgu}\\   
        \Lf{i} \dot{i}_{\text{f},i} & = v_{\text{t},i} -\Rf{i} \ir{i} - v_i ,
    \end{align}
\end{subequations}
where $i_{\text{ext},i}$ is the cumulative current injected by interconnecting lines. The load is modeled with a voltage-dependent current source 
\begin{align} \label{eq:load}
	\iL{i}(v_i) = y_i v_i + \frac{p_i}{v_i} + \hat{i}_{\text{L},i},
\end{align} 
where $y_i \in \R_{\geq0}$ is the admittance of a constant resistive load, $p_i \in \R$ is a constant power load, and $\hat{i}_{\text{L},i} \in \R$ is a constant current load. Thus, the load is a linear combination of a constant resistive, a constant power and constant current load, which is also known as a nonlinear ZIP load \cite{kundur2022power}.  
\begin{figure}[h]
	\centering
	\begin{tikzpicture}

\node[draw, align=center, minimum width = 1.5cm, minimum height=2cm] (LE) at (0,0) {Buck \\ converter};
\draw (2,0.75) -- ($(LE.east)+(0,0.75)$); 
\draw (2,-0.75) -- ($(LE.east)+(0,-0.75)$); 
\draw[fill=white] (2,0.75) circle (.1cm);
\draw[fill=white] (2,-0.75) circle (.1cm);

\node[draw, fill=black, minimum width=.6cm] (Lf) at (3,0.75) {};
\node[above] (Lf_name) at (3,0.8) {$L_{{\rm f},i}$};
\node[draw, fill=white, minimum width=.6cm] (Rf) at (4,0.75) {};
\node[above] (Rf_name) at (4,0.8) {$R_{{\rm f},i}$};
\node[draw,minimum width=0.6cm, minimum height=0.3cm] (Cf) at (5,0) {};
\node[fill=white,minimum width=0.7cm,minimum height=0.29cm] at (5,0) {};
\node at (4.4,0) {$C_{{\rm f},i}$};

\node[draw, circle, fill=white, minimum height=.5cm] (Rl) at (6,0) {};
\draw[-latex] ($(Rl.north)+(0,-0.1)$) --  ($(Rl.south)+(0,0.1)$);
\draw[-latex] ($(Rl.north)+(0,0.5)$) --  ($(Rl.north)+(0,0.2)$) node[right] {$i_{{\rm L},i}$};

\draw (2.1,0.75) -- (Lf);
\draw (Lf) -- (Rf);
\draw (Rf) -| (Cf);
\draw (2.1,-0.75) -| (Cf);
\draw ($(Cf)+(0,0.75)$) -| (Rl);
\draw ($(Cf)+(0,-0.75)$) -| (Rl);

\draw[fill=white] (7,0.75) circle (.1cm);
\draw[fill=white] (7,-0.75) circle (.1cm);
\draw (6.9,0.75) -- ($(Rl)+(0,0.75)$); 
\draw (6.9,-0.75) -- ($(Rl)+(0,-0.75)$); 
\draw[-latex] (6.8,0.75) -- (6.4,0.75) node[above] {$i_{{\rm ext},i}$};


\draw[-latex] (1.5,0.6)  -- node[left] {$v_{{\rm t},i}$} (1.5,-.6);
\draw[-latex] (1.5,0.75)  -- node[above] {$i_{{\rm f},i}$} (1.8,0.75);
\draw[fill=black, above] (5,0.75) node {$v_i$} circle (.05cm);

\end{tikzpicture}
	\caption{Electric scheme of a DGU and load in node $i \in \setOfnodes_{\text{dgu}}$.}
	\label{fig:dgu}
\end{figure}
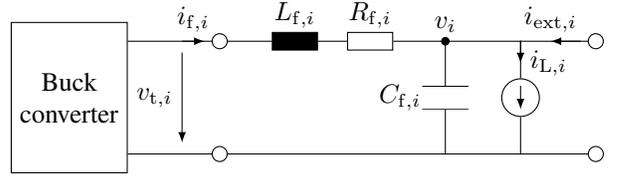

\subsection{Grid-forming and grid-following control}

The \gls{dgu} described in Section~\ref{sec:dgu} is normally equipped with a grid-forming or a grid-following controller \cite{rocabert2012control}. Grid-forming controllers inject the necessary filter current $\ir{i}$ (indirectly power) in order to regulate the node voltage $v_i$ to a desired voltage reference $v_\reff$, and stabilize thus the grid-voltages regardless of the load disturbance or volatile power injections. Grid-following controllers set the voltage $v_{\text{t},i}$ such that a given power reference $p_\reff$ (indirectly the filter current $\ir{i}$) is injected, without considering the resulting node voltage level $v_i$. Grid-forming \glspl{dgu} are used to achieve robust voltage stability, and grid-following \glspl{dgu} to inject a certain amount of power irrespective of grid stability, e.g. for achieving optimal dispatch. Therefore, in this work, we have exactly one grid-forming \gls{dgu} in every microgrid, and an arbitrary number of grid-following \glspl{dgu}. The first \gls{dgu} $1 \in \setOfnodes_{\text{DGU}}$ is defined without loss of generalization as the grid-forming \gls{dgu}, while ${i \in \setOfnodes_{\text{DGU}}\setminus\{ 1 \} }$ are grid-following \glspl{dgu}. 

The grid-forming controller is taken from \cite{nahata2020passivity} and is designed by introducing an error state~\eqref{eq:error_vref} and a state feedback as
\begin{subequations} \label{eq:gfr}
	\begin{align}
		v_{\text{t},1} & = k_{\alpha,1} v_1 + k_{\beta,1} \ir{1} + k_{\gamma,1}  e_1 \\
		\dot{e}_1 & = v_{\reff,1} - v_1, \label{eq:error_vref}
	\end{align}
\end{subequations}
where $k_{\alpha,1}\in \R$, $k_{\beta,1}\in \R$ and $k_{\gamma,1} \in \R$ are the controller parameters. Similarly, the grid-following controller is designed here as 
\begin{subequations} \label{eq:gfl}
	\begin{align}
		v_{\text{t},i} & = k_{\alpha,i} v_i + k_{\beta,i} \ir{i} + k_{\gamma,i}  e_i \\
		\dot{e}_i & = p_{\reff,i} - v_i \ir{i}, \label{eq:error_pref}
	\end{align}
\end{subequations}
for all $i \in \setOfnodes_{\text{dgu}} \setminus \{1\}$, using the injected power error~\eqref{eq:error_pref} instead. Note that the grid-following \gls{dgu} introduces a nonlinearity when computing the injected power $v_i \ir{i}$ in \eqref{eq:error_pref}. Applying either~\eqref{eq:gfr} or~\eqref{eq:gfl} to the \gls{dgu}~\eqref{eq:dgu} thus yields
\begin{subequations}
    \begin{align}
        \Cf{i} \dot{v}_{i} & = \ir{i} - \iL{i}(v_i) - i_{\text{ext},i} \\   
        \dot{i}_{\text{f},i} & = \alpha_i v_i + \beta_i \ir{i} + \gamma_i e_i \\
		\eqref{eq:error_vref} & \; \text{or} \; \eqref{eq:error_pref},   
    \end{align}
\end{subequations}
depending if it is a grid-forming or grid-following \gls{dgu}, with the variables $\alpha_i = \frac{k_{\alpha,i}-1}{\Lf{i}}$, $\beta_i = \frac{k_{\beta,i}-\Rf{i}}{\Lf{i}} $ and $\gamma_i = \frac{k_{\gamma,i}}{\Lf{i}}$ containing the controller parameters.  

\subsection{Load node}

A load node is modeled as a \gls{dgu} node without the filter and Buck converter. The dynamics are governed by the nonlinear dynamic equation 
\begin{align}\label{eq:voltage_dynamics_load}
	\Cf{i} \dot{v}_{i} & = - \iL{i}(v_i) - i_{\text{ext},i} \qquad \forall i \in \setOfnodes_{\text{L}},
\end{align}
where $\iL{i}(v_i)$ is the load as in \eqref{eq:load}.

\subsection{Power lines}

The power lines $j \in \mathcal{E}$ interconnecting buses $i,h \in \setOfnodes$ are modeled with the $\pi$-equivalent circuit~\cite{machowski2020power}, as shown in Fig.~\ref{fig:powerline}. It is composed of a series inductance $\Lpi{j} \in \R_{>0}$ and resistance $\Rpi{j}\in \R_{>0}$, and two parallel capacitances $\frac{\Cpi{j}}{2}\in \R_{>0}$. 
\begin{figure}[h]
	\centering
	\includegraphics{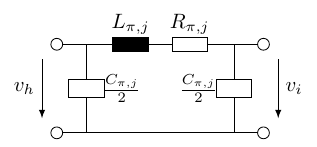}
	\caption{Electric scheme of a power line in the pi equivalent circuit.}
	\label{fig:powerline}
\end{figure}
Note that the line capacitance is connected in parallel to the bus filter capacitance of the bus $i$ or $h \in \setOfnodes$. Thus, the line capacitances are added to the \gls{dgu} filter capacitances. The line dynamics are thus described by 
\begin{align}\label{eq:line_dynamics} 
    \Lpi{j} \dot{i}_{\pi,j} = -\Rpi{j} i_{\pi,j} + v_{\Delta,j},
\end{align}
where $v_{\Delta,j} = v_i - v_h$ is the input and $v_i$, $v_j$ are the bus voltages.

\subsection{Microgrid Model}

The whole microgrid is composed of $\numOfnodes = \vert \setOfnodes \vert$ electrical buses with $\numOfDGU = \vert \setOfnodes_{\text{DGU}} \vert \geq 1$ \glspl{dgu}, interconnected by $\numOflines = \vert \setOfedges \vert$ power lines according to the graph $\mathcal{G}(\setOfnodes,\mathcal{E})$ with incidence matrix $M\in\R^{\numOfnodes\times \numOflines}$. With respect to the $\numOfDGU$ \glspl{dgu}, recall that there is one grid-forming \gls{dgu} which stabilizes the grid voltages and $\numOfDGU-1$ grid-following \glspl{dgu} that inject power according to an economic objective function, which will be specified in Section~\ref{sec:controller}. The microgrid model reads
\begin{subequations} \label{eq:microgrid_dynamics}
    \begin{align}
        \Cf{} \dot{v} & = I_{\text{f}} \ir{} - \iL{}(v) - \inci i_\pi \label{eq:eq1} \\   
        \dot{i}_{\text{f}} & = \alpha I_{\text{f}}^\top v + \beta \ir{} + \gamma e \\
        \dot{e} & = I_v v + I_p [I_{\text{f}}^\top v] \ir{} + \mat{ v_{\text{ref}} \\ p_{\text{ref}}} \\
        \Lpi{} \dot{i}_\pi &= -\Rpi{} i_{\pi} + \inci^T v, \label{eq:line_dynamics_1}
    \end{align}
\end{subequations}
where $\alpha = [\alpha_i]$, $\beta = [\beta_i]$ and $\gamma = [\gamma_i]$ contain the control parameters, $\Cf{} = [\Cf{i}]$, $\Rf{}= [\Rf{i}]$, $\Lf{} = [\Lf{i}]$, $\iL{} = [\iL{i}]$, $\Rpi{} = [\Rpi{j}]$ and $\Lpi{} = [\Lpi{j}]$ are the filter, load and line parameters, and $v = \col{v_i} \in \R^\numOfnodes$, $\ir{} = \col{\ir{i}} \in \R^\numOfDGU$, $e = \col{e_i} \in \R^{\numOfDGU}$ and $i_\pi = \col{i_{\pi,j}}$ are the stacked states of the \glspl{dgu} $i \in \setOfnodes$ and power lines $j \in \setOfedges$. The voltage and power references $v_\reff \in \R_{>0}$ and $p_\reff \in \R^{\numOfDGU-1}$ are inputs, where $\numOfDGU$ defines the number of inputs of the microgrid. The matrix $I_{\text{f}} \in \R^{\numOfnodes\times \numOfDGU}$ is a permutation matrix assigning the filter currents of $\numOfDGU$ \glspl{dgu} to the correct $\numOfnodes \geq \numOfDGU$ nodes. The matrices $I_v = \diag{1,0,\dots,0}\in \R^{\numOfDGU\times \numOfDGU}$ and $I_p = \diag{0,1,\dots,1} \in \R^{\numOfDGU\times \numOfDGU}$ are diagonal matrices such that the correct error signals are induced for the integrator states\footnote[1]{It holds that $I_v+I_p = I$, since all DGUs are either grid-forming or grid-following.} as in~\eqref{eq:error_vref} and~\eqref{eq:error_pref}. Note that in~\eqref{eq:eq1} and~\eqref{eq:line_dynamics_1} it has been taken into account that the voltage drop over the power lines $v_\Delta$ can be expressed as $v_\Delta = \inci^\top v$ and the current drawn from the buses by the power lines as $i_{\text{ext}} =  \inci i_\pi$. The output $y$ of the system is defined as the filter current of the grid-forming \gls{dgu}, which is set as the first filter current without loss of generality, i.e. $y=\ir{1}$. 

In the next section, a price-based controller is designed in order to determine $p_{\text{ref}}$ such that an optimal steady state is achieved. 

\section{Price-based controller design} \label{sec:controller} 

We aim to design a controller which (i) steers the microgrid~\eqref{eq:microgrid_dynamics} to an (unknown) economically optimal operation point where the grid-forming \gls{dgu} does not inject power, and (ii) has distributed nature and does not require any knowledge of loads or line parameters. 

\subsection{Controller design: Optimality model}
Inspired by the Linear-Convex Optimal Steady-State Control \cite{lawrence2020linear}, we introduce an optimality model, which describes an optimal steady state where property (i) is fulfilled:  
\begin{subequations} \label{eq:optimization_model}
    \begin{align}
        \min_{p_\reff} \; &  \sum_{i=1}^{\numOfDGU-1} f_i(p_{\reff,i}) \\
        \text{s.t.} \; & \sum_{i=1}^{\numOfDGU-1} p_{\reff,i} = \Pload. \label{eq:power_equality}
    \end{align}
\end{subequations}
The function $f_i: \R \rightarrow \R$ represents the cost of the power infeed of the respective grid-following \glspl{dgu}, which is assumed to be convex and quadratic in the paper at hand, i.e. $f(p_{\reff,i}) = q_i p_{\reff,i}^2 + r_i p_{\reff,i} + s_i$ with $q_i, r_i, s_i \in \R$, $q_i > 0$. The variable $\Pload$ comprises the sum of the power consumed by all loads and the losses of the microgrid. Thus, \eqref{eq:power_equality} ensures power balance. 
The KKT conditions \cite{wright1999numerical} for~\eqref{eq:optimization_model} are 
\begin{subequations}
	\begin{align}
		0 &= \nabla f_i(p_\reff) + \lambda \qquad \forall i \in \{1,\dots,d-1\} \\ 
		0 &= \sum_{i=1}^{\numOfDGU-1} p_{\reff,i} - \Pload,
	\end{align}
\end{subequations}
where $\lambda \in \R$ is the Lagrange multiplier for the constraint \eqref{eq:power_equality}, and using a primal-dual gradient method \cite{arrow1958studies,feijer2010stability} with positive tuning parameters $\tau_i$ and $\kappa$, we get
\begin{subequations} \label{eq:controller}
	\begin{align}
		\dot{p}_{\reff,i} &= - \tau_i(\nabla f_i(p_{\reff,i}) - \lambda) \quad \forall i \in \{1,\dots,d-1\} \label{eq:marginal_cost} \\ 
		\dot{\lambda} &= \kappa(\Pload - \sum_{i=1}^{\numOfDGU-1} p_{\reff,i}) . \label{eq:price_forming}
	\end{align}
\end{subequations}
The multiplier $\lambda$ can be interpreted as the electrical power price; if the load $\Pload$ is greater than the power supplied by the grid-following \glspl{dgu}, the price in~\eqref{eq:price_forming} increases and vice-versa. Equation~\eqref{eq:marginal_cost} means that the grid-following \glspl{dgu} inject power such that their marginal costs equal the power price. This is the best solution for rational decision-makers, since feeding in more power would lead to less economic benefit per kW.  

With controller \eqref{eq:controller}, property (i) is fulfilled, since at steady state, every grid-following \gls{dgu} produces at marginal cost and the grid-forming \glspl{dgu} inject no power. However, property (ii) requires more attention. Even if~\eqref{eq:marginal_cost} can be computed by every grid-following \gls{dgu} in distributed manner, the price-forming~\eqref{eq:price_forming} uses the load $\Pload$ and the sum of the power injection $p_{\reff,i}$ of the grid-following \glspl{dgu}, both being system-wide knowledge. To circumvent that, we present the following proposition.
\begin{lemma} \label{lem:power_gfr}
    Let $v_{\text{ref}} \in \R_{>0}$. Let all the load in microgrid~\eqref{eq:microgrid_dynamics}, including the loads and transmission losses, be denoted by $\Pload$. Then, at steady state, we have $y=0$ iff
    \begin{equation} \label{eq:gfl_load}
        \sum_{i=1}^{\numOfDGU-1} p_{\reff,i} = \Pload. 
    \end{equation}
\end{lemma}
\begin{proof}
    In any steady state with $v_\reff\in\R_{>0}$ and $p_\reff \in\R^{\numOfDGU-1}$, it holds from~\eqref{eq:microgrid_dynamics}
    \begin{subequations}
        \begin{align}
            0 &= I_{\text{f}} \ir{} -\iL{}(v) - \inci i_\pi \label{eq:ss_voltage} \\
            0 &= \alpha I_{\text{f}}^\top v + \beta \ir{} + \gamma e \\
            0 &= v_\reff - v_1 \\
            0 &= p_\reff - I_p [I_{\text{f}}^\top v] \ir{} \label{eq:power_error}\\
            0 &= \inci^\top v - \Rpi{} i_\pi. \label{eq:ss_powerline}
        \end{align}    
    \end{subequations}
    Rearranging \eqref{eq:ss_powerline} to $i_\pi = R^{-1}M^\top v$ and inserting it in \eqref{eq:ss_voltage}, we have
    \begin{equation}\label{eq:ss_voltage_1}
        0 = I_{\text{f}} \ir{} - \iL{}(v) - \inci \Rpi{}^{-1} \inci^\top v,
    \end{equation}
    and multiplying \eqref{eq:ss_voltage_1} with $v^\top$ from the left, we have
    \begin{equation}\label{eq:ss_voltage_2}
        0 = v^\top I_{\text{f}} \ir{}  \underbrace{- v^\top \iL{}(v) - v^\top \inci \Rpi{}^{-1}\inci^\top v}_{-\Pload}. 
    \end{equation}
    Taking into account that the last two terms in \eqref{eq:ss_voltage_2} denote the sum of all loads and the power line losses $\Pload$, \eqref{eq:ss_voltage_2} simplifies to
    \begin{equation}\label{eq:ss_voltage_3}
        v^\top I_{\text{f}} \ir{} = \Pload. 
    \end{equation}
	Sum over~\eqref{eq:power_error} to obtain $v^\top I_{\text{f}} I_p \ir{} = \sum_i p_{\reff,i}$, and note that $v^\top I_{\text{f}} \ir{} = v^\top I_{\text{f}} I_p \ir{} + v^\top I_{\text{f}} I_v \ir{}$. Insert both equations in the left-hand side of~\eqref{eq:ss_voltage_3}. Then, for $v_1 = v_\reff >0$, \eqref{eq:gfl_load} follows iff $y = \ir{1} = 0$, since $v^\top I_{\text{f}} I_v \ir{} = v_1 \ir{1}$.
\end{proof}

Applying Proposition~\ref{lem:power_gfr} to \eqref{eq:controller}, we can use the grid-forming \gls{dgu} current $y = \ir{1}$ for the price-forming mechanism, since it is a measure for the unmet power demand in the microgrid, i.e. 
\begin{subequations} \label{eq:controller_distributed}
	\begin{align}
		\dot{p}_{\reff,i} &= - \tau_i(\nabla f_i(p_\reff) - \lambda) \qquad \forall i \in \{1,\dots,d-1\} \label{eq:marginal_cost_1} \\ 
		\dot{\lambda} &= -\kappa y . \label{eq:price_forming_distributed}
	\end{align}
\end{subequations}
This way, the price-forming mechanism does not need system-wide knowledge; the price is formed solely by the grid-forming \gls{dgu} and forwarded to the grid-following \glspl{dgu} within the microgrid\footnote{Note that it may be forwarded in a distributed manner}. Thus, the \gls{dgu} responsible for stabilizing the grid (grid-forming) is the price-making entity, and the grid-following \glspl{dgu} are the price-taking agents. 

\subsection{Closed-loop stability}

In this subsection, we study the stability of the microgrid~\eqref{eq:microgrid_dynamics} controlled with~\eqref{eq:controller_distributed}. Define the state variable  $x = \col{v, \ir{}, e, i_\pi, p_\reff, \lambda} \in \R^{\numOfstates}$ with $\numOfstates = \numOfnodes+ 3\numOfDGU + \numOflines$ the number of states, and let $\rl{x}$ be an equilibrium point of the closed-loop system for a constant $v_\reff$. In shifted coordinates $\err{x} = x - \rl{x}$, the nonlinear closed-loop system reads
\begin{align} \label{eq:sys_shifted}
	\dot{\err{x}} = A(\err{x},\rl{x},P) \err{x}  
\end{align} 
with 
{\small
\begin{align} 
	&A(\err{x},\rl{x},P) = \\
	&\quad \mat{- \Cf{}^{-1} Y+P[\bar{v}]^{-1}[v]^{-1} & \Cf{}^{-1} I_{\text{f}} & 0 & -\Cf{}^{-1} \inci & 0 & 0 \\ \alpha I_{\text{f}}^\top & \beta & \gamma & 0 & 0 & 0 \\ -I_v - I_p[\ir{}] & -I_p [\bar{v}] & 0 & 0 & {I}_p & 0 \\ \Lpi{}^{-1}\inci & 0 & 0 & -\Lpi{}^{-1} \Rpi{} & 0 & 0 \\ 0 & 0 & 0 & 0 & -Q & -\tau \\ 0 & -\kappa & 0 & 0 & 0 & 0},
\end{align} 
}
where $v = \rl{v} + \err{v}$, $\ir{} = \rl{i}_{\text{f}} + \err{i}_{\text{f}}$, $Q = [\tau_i q_i]$, $\tau = \col{\tau_i}$ and $Y = [y_i]$. The constant power load $P = \col{p_i} \in \R^\numOfnodes$ is seen as a variable here, since we consider it to vary over time. 

\begin{definition}
	The feasible subspace of the state space $\X \subset \R^{\numOfstates}$ for safe operation is defined as
\begin{align}
	\X = \V \times \I \times \R^{2\numOfDGU + \numOflines}, 
\end{align}
where $\V = [v_{\min}, v_{\max}] \subset \R^{\numOfnodes}$ and $\I = [i_{\text{f},\min}, i_{\text{f},\max}] \subset \R^\numOfDGU$ are polytopic sets describing maximum and minimum feasible node voltages and filter currents. 
\end{definition}


In the following, we present a condition to assess stability of the closed-loop system by solving a semi-definite program.

\begin{lemma} \label{th:stability}
	Any equilibrium point $\rl{x}$ with $\rl{v} \in {\V}$ and $\rl{i}_{\text{f}}  \in {\I}$ is asymptotically stable if there exists a symmetric $S \in \R^{\numOfstates \times \numOfstates}$ such that
	\begin{align}
		S > 0 \label{eq:lyapunov_microgrid_1} \\
		A(\err{x},\rl{x},P)^\top S + S A(\err{x},\rl{x},P)  < 0, \label{eq:lyapunov_microgrid_2}
	\end{align}
	hold for all steady-states $\rl{v} \in {\V}$, $\rl{i}_{\text{f}} \in {\I}$ and $P \in \Pl$, and for the subset of the state space $\err{v} \in \err{\V}$, $\err{i}_{\text{f}}  \in \err{\I}$, where $\err{\V}=[\err{v}_{\min}, \err{v}_{\max}]\subset \R^{\numOfnodes}$ and $\err{\I} = [\err{i}_{\text{f},\min}, \err{i}_{\text{f},\max}]\subset \R^{\numOfDGU}$ are sets describing the maximal deviation of node voltages and filter currents, and $\Pl = [p_{\min}, p_{\max}]\subset \R^{\numOfnodes}$ the constant power load bounds.
\end{lemma}
\begin{proof}
	Denote $\err{X} = \err{\V} \times \err{\I} \times \R^{2\numOfDGU+\numOflines}$. Consider the Lyapunov function $V: \err{X} \rightarrow \R$ with $V(\err{x}) = \err{x}^\top S \err{x}$. With~\eqref{eq:lyapunov_microgrid_1}, the positive definiteness $V>0$ in $\err{X}-\{0\}$ is ensured. The time derivative of the Lyapunov function is $\dot{V} = \err{x}^\top\big( A(\cdot)^\top S + S A(\cdot) \big) \err{x}$. The linear matrix inequality~\eqref{eq:lyapunov_microgrid_2} ensures that the time derivative of the Lyapunov function is negative in $\err{X} - \{0\}$ for any possible steady-state $\rl{v} \in {\V}$, $\rl{i}_{\text{f}} \in {\I}$. Thus, there exists a Lyapunov function fulfilling the conditions in \cite[Theorem~4.1]{khalil2002nonlinear} for all possible steady states, which implies asymptotically stability for any steady state. 
\end{proof}
\begin{remark}
	Computing a matrix $S$ that fulfills \eqref{eq:lyapunov_microgrid_1} and \eqref{eq:lyapunov_microgrid_2} $\forall \err{x} \in \err{X}$, $\forall \rl{x} \in {X}$ and $\forall P \in \Pl$ is a semi-definite program if the sets $\err{X}, \rl{X}$ and $\Pl$ are convex polytopes. Then, matrix $S$ has to satisfy \eqref{eq:lyapunov_microgrid_1} and \eqref{eq:lyapunov_microgrid_2} for all the vertices of the convex hull of $\err{X}, {X}$ \cite[Ch.~5.1]{boyd1994linear}. 
\end{remark}
%
%

Proposition~\ref{th:stability} above gives a condition which is easy to check via semi-definite programming for the closed-loop stability of any feasible equilibrium of~\eqref{eq:sys_shifted}. 

In the next section, the interconnection ports for considering networked microgrids are defined, and the dissipativity properties are studied. 

\section{Interconnection of microgrids with price-based controllers} \label{sec:ports}

The following definitions lay the foundation for analysing the stability and optimality of a set of networked microgrids~\eqref{eq:microgrid_dynamics}, each controlled with~\eqref{eq:controller_distributed}. First, we define novel physical and economic interconnection ports and study their dissipativity properties. Then, we propose an interconnection scheme for the microgrids $k \in \M$ such that the networked microgrids are asymptotically stable and operate at a globally economically optimal steady state. 

Since an interconnection of multiple microgrids is considered, the microgrid index $k \in \M$ is not further omitted.

\subsection{Microgrid ports: Definition and dissipativity}

The following electric port defines an interface for interconnecting microgrids via electric lines.
\begin{definition}[electric ports] \label{def:elec}
	Let $i_{\ext,i}^k$ be an external current injected at a node $i \in \setOfnodes$ and $v_i^k$ the voltage at that node for microgrid $k \in \M$. The input-output pair $(i_{\ext,i}^k, v_i^k)$ is called an electric port\footnote{Note that electric ports have been used in the literature for interconnecting \glspl{dgu} and lines~\cite{nahata2020passivity,strehle2020} \emph{within} a microgrid. Definition~\ref{def:elec} can hence be understood as leveraging these ports \emph{between} microgrids.} for that microgrid. 
\end{definition}

The electric port is interfaced with system~\eqref{eq:sys_shifted} through the vectors $b_\ext^k = \col{t_i, 0_{3\numOfDGU+\numOflines}}$ and $c_\ext^k = b_\ext^{k\top}$, where $t_i \in \R^\numOfnodes$ has a 1 at the $i$-th element and zero elsewhere, since an external current drawn to a node $i \in \setOfnodes$ acts on the voltage dynamics~\eqref{eq:voltage_dynamics_dgu} or~\eqref{eq:voltage_dynamics_load} of node $i \in \setOfnodes$. Note that a microgrid may contain an arbitrary number $z \in \N$ of electric ports, yielding matrices $B_\ext^k = [b_{\ext,1}^k,\dots,b_{\ext,z}^k]$ and $C_\ext^k = B_\ext^{k\top}$.

The following economic port defines an interface for interconnecting microgrids economically. 
\begin{definition}[economic ports]
	Let $\lambda_{\glob}^k \in \R_{}$ denote an external electric power price and $\lambda_{\loc}^k \in \R_{}$ the local price for a certain microgrid. The input-output pair $(\lambda_{\glob}^k, \lambda_{\loc}^k)$ is called the economic port for microgrid $k \in \M$. 
\end{definition}

When the economic port $(\lambda_{\glob}^k, \lambda_{\loc}^k)$ is connected, we replace the price used for the grid-following \glspl{dgu} in~\eqref{eq:marginal_cost_1} with the input $\lambda_\glob^k$, yielding
\begin{subequations} \label{eq:ports_econ}
	\begin{align}
		\dot{p}_{\reff,i}^k &= - \tau_i^k(\nabla f_i^k(p_\reff^k) - \lambda_\glob^k) \label{eq:pref_externalPrice} \\ 
		\dot{\lambda}_\loc^k &= -\kappa^k \ir{1}^k . \label{eq:local_price_forming}
	\end{align}
\end{subequations}
The local price $\lambda_\loc^k$ (output of economic port) is still determined by the current of the grid-forming \gls{dgu}, but is no longer used directly in the local microgrid. Splitting the price in a microgrid into local and external prices allows, using a special interconnection structure for economic ports as proposed in Section~\ref{sec:interconnection_econ}, the local price $\lambda_\loc^k$ to contribute towards a (global) external price. The external price then already implicitly contains a coordination between microgrids, and is used by the grid-following \glspl{dgu} in order to achieve global optimal dispatch. Note that only a single economic port per microgrid is allowed in this work, since we have a single local price per microgrid.

The economic port thus interfaces with the system~\eqref{eq:sys_shifted} through the vectors $b_\econ^k = \col{0_{\numOfnodes+2\numOfDGU+\numOflines}, \tau, 0}$ and $c_\econ^k = \col{0_{\numOfnodes+3\numOfDGU+\numOflines-1}, 1}$.
System~\eqref{eq:sys_shifted} with electric and economic ports reads then
\begin{subequations} \label{eq:sys_with_ports}
	\begin{align} 
		\dot{\err{x}}^k & = A_\lambda^k(\err{x}^k,\rl{x}^k,P^k) \err{x}^k + B_\ext^k i_{\ext}^k + b_\econ^k \lambda_\glob^k \\
		\err{y}_\ext^k &= C_\ext^k \err{x}^k = v_\ext^k \\
		\err{y}_\econ^k &= c_\econ^k \err{x}^k = \lambda_\loc^k.
	\end{align} 
\end{subequations}
The matrix $A_\lambda^k(\cdot)$ is the same as $A(\cdot)$ in \eqref{eq:sys_shifted} except that~\eqref{eq:pref_externalPrice} replaces~\eqref{eq:marginal_cost_1}, where the input $\lambda_\glob^k$ is used instead of the state from~\eqref{eq:local_price_forming}. Vector $i_\ext^k = \col{i_{\ext,z}^k}$ is the input for all electric ports $z$. 

We are now interested in the dissipativity properties of the interconnection ports, in order to analyse microgrids interconnected via physical-electric or information-economic ports. First, we analyse the dissipativity properties of the electric port; thereafter, the properties of the economic port. 
\begin{lemma} \label{lem:electric_ports}
	Let a microgrid self-close its economic port with $\lambda_{\glob}^k = y_{\econ}^k = \lambda_\loc^k$, i.e. without interconnecting economically with other microgrids. System~\eqref{eq:sys_with_ports} is then equilibrium independent passive \cite{hines2011equilibrium} w.r.t. the electric port $(v_\ext^k, i_{\ext}^k)$ if a there exists a symmetric $S \in \R^{\numOfstates\times \numOfstates}$ solving
	\begin{subequations}
		\begin{align}
			S & > 0 \\ 
			\mat{A_\lambda^k(\cdot)^\top S + S A_\lambda^k(\cdot) & S B_\ext^k - C_\ext^{k\top} \\ B_\ext^{k\top} S - C_\ext^k & 0 } & \leq 0 \label{eq:passivity_electric_port}
		\end{align} 
	\end{subequations}
	for all $\rl{x}^k \in {\X}$, $\err{x}^k \in \err{\X}$ and $P^k \in \Pl$, with $\X$, $\err{\X}$ and $\Pl$ defined as in Proposition~\ref{th:stability}.
\end{lemma}
\begin{proof}
	Consider the storage function $V(\err{x}^k) = \err{x}^{k\top} S \err{x}^k$ with $S>0$. The passivity condition $\dot{V} < \err{v}_\ext^{k\top} \, \err{i}_{\ext}^k$ leads to~\eqref{eq:passivity_electric_port} for the system with self-closed economic port, i.e. $\lambda_\glob^k = y_\econ^k$. 
\end{proof}

Proposition~\ref{lem:electric_ports} ensures stability of a scenario where different microgrids are interconnected via electric ports. Since an economic port is not considered, the electric power prices in the microgrids are independent, yielding optimal operation in each microgrid but suboptimal operation of the networked microgrids as is shown in Section~\ref{sec:simulation}. 

In order to interconnect the microgrids via the economic port and achieve an economic cooperation, we study the dissipativity properties of both port types simultaneously. 
\begin{theorem} \label{th:econ}
	System~\eqref{eq:sys_with_ports} is input-feedforward and output-feedback equilibrium independent passive w.r.t. the electric $(v_\ext^k, i_{\ext}^k)$ and economic ports $(\lambda_{\glob}^k, \lambda_{\loc}^k)$ if there exists a symmetric $S \in \R^{\numOfstates\times \numOfstates}$, and indices $\nu^k \in \R$, $\rho^k \in \R$ such that 
	{\small%
	\begin{subequations}
		\begin{align} 
			S & > 0 \\ 
			\mat{ X(\err{x}^k,\rl{x}^k,\rho^k) & S B_\ext^k - C_\ext^{k\top} &  S b_\econ^k - c_\econ^{k\top} \\ B_\ext^{k\top}  S - C_\ext^k & 0 & 0 \\ b_\econ^{k\top} S - c_\econ^k & 0 & \nu^k } & \leq 0 \label{eq:passivity_economic_port}
		\end{align} 
	\end{subequations}
	}
	holds for all $\rl{x}^k \in \rl{\X}$, $\err{x}^k \in \err{\X}$ and $P^k \in \Pl$, where {$X(\err{x}^k,\rl{x}^k,\rho^k) = A_\lambda^k (\cdot)^\top S + S A_\lambda^k(\cdot) + \rho^k c_\econ^{k\top} c_\econ^k$}.
\end{theorem}
\begin{proof}
	Consider the storage function $V(\err{x}^k) = \err{x}^{k\top} S \err{x}^k$ with $S>0$. The passivity condition $\dot{V} < \err{v}_\ext^{k\top} \, \err{i}_{\ext}^k + \err{\lambda}_{\loc}^k \err{\lambda}_{\glob}^k - \nu^k (\err{\lambda}_\glob^k)^2 - \rho^k (\err{\lambda}_\loc^k)^2$ leads to~\eqref{eq:passivity_economic_port} when using~\eqref{eq:sys_with_ports} in the derivative of the storage function. The indices $\nu^k$ and $\rho^k$ indicate the excess or lack of passivity \cite{sepulchre2012constructive}.  
\end{proof}
\begin{remark}
	Depending on the loads and DGUs of a microgrid, the indices $\nu^k$ and $\rho^k$ may be negative indicating a lack of passivity of the economic port, which needs to be compensated as proposed in Section~\ref*{sec:interconnection_econ}.
\end{remark}

In this subsection, the interconnection ports have been defined and their passivity properties studied. Next, an interconnection scheme for the electric and economic ports of networked microgrids ensuring global economic optimality and asymptotic stability is proposed.

\subsection{Networked microgrid operation} \label{sec:interconnection_econ}

The interconnection of microgrids with electric ports through electric lines is analogous to the interconnection of nodes via electric lines as explained in Section~\ref{sec:system}. Given a connected interconnection topology, the interconnection of nodes via electric lines constitutes a skew-symmetric interconnection as shown in \cite{nahata2020passivity}. Thus, microgrids with electric ports fulfilling Proposition~\ref{lem:electric_ports} or Theorem~\ref{th:econ} interconnected via electric lines are guaranteed to be asymptotically stable.   

We now consider a set $\M$ of microgrids interconnected electrically and economically, using the port properties established in Theorem~\ref{th:econ}. To this end, we first characterize a globally optimal steady state for economically interconnected microgrids. 
\begin{lemma} \label{lem:optimality}
	If the external electric power price $\lambda_{\glob}^k \in \R_{}$ is equal for all microgrids, i.e.  
	\begin{align} \label{eq:optimal_dispatch_condition}
		\lambda_{\glob}^k = \rl{\lambda}, \qquad \forall k \in \M
	\end{align}
	with constant $\rl{\lambda} \in \R_{>0}$, we have global optimal dispatch and all \glspl{dgu} inject power at marginal cost at any steady state.
\end{lemma}
\begin{proof}
	From \eqref{eq:pref_externalPrice} it follows that all \glspl{dgu} in a specific microgrid operate at marginal cost if $\lambda_\glob^k$ is constant. Furthermore, at steady state, from~\eqref{eq:local_price_forming}, the grid-forming \gls{dgu} does not inject power and all the loads are supplied optimally by the grid-following \glspl{dgu}. In addition, if~\eqref{eq:optimal_dispatch_condition} holds, we have the same price (and hence equal marginal costs) at every microgrid yielding global optimal dispatch among all grid-following \glspl{dgu}. 
\end{proof}

To achieve $\lambda_{\glob}^k = \rl{\lambda}$ $\forall k$ as in Proposition~\ref{lem:optimality}, we propose a consensus-based algorithm with which the microgrids perform a distributed dynamic averaging of the local price $\lambda_\loc^k$ (output of the economic port). The output of the distributed dynamic averaging is used as the external price $\lambda_\glob^k$ (input of the economic port). Then, at steady state, the external prices $\lambda_\glob^k$ of all $\numOfMG$ microgrids taking part in the distributed dynamic averaging are equal, i.e. $\lambda_\glob^k = \frac{1}{\numOfMG}\sum_{k=1}^{\numOfMG} \lambda_{\loc}^k$, and Proposition~\ref{lem:optimality} is fulfilled. There exist many dynamic consensus algorithms, see \cite{kia2019tutorial} for a survey. In this work, a dynamic consensus algorithm that enjoys an excess of passivity is needed to compensate the lack of passivity of the economic port (characterized in Theorem~\ref{th:econ} via $\nu^k$ and $\rho^k$). We therefore use the proportional dynamic consensus~\cite{kia2019tutorial}
\begin{subequations} \label{eq:PDDA}
	\begin{align}
		\dot{w} &= -(\mu I_{\numOfMG} + L) w - L \lambda_\loc \\
		\lambda_\glob & = w + \lambda_\loc,
	\end{align}
\end{subequations}
where $L\in \R^{\numOfMG\times \numOfMG}$ is the Laplacian matrix of an arbitrary but connected topology describing the communication between the microgrids via economic ports, $\mu \in \R_{>0}$ a tuning parameter and $w \in \R^{\numOfMG}$ auxiliary states. Note that the input $\lambda_\loc = \col{\lambda_\loc^k}$ and output $\lambda_\glob = \col{\lambda_\glob^k}$ of consensus algorithm~\eqref{eq:PDDA} correspond to the economic port as described. All local prices thus contribute to the global, external price.
This consensus protocol is chosen because it exhibits an excess of input and output passivity (it is input-to-state stable~\cite[Theorem~3]{freeman2006stability} and has feedthrough, both related to an excess of passivity~\cite{van2000l2}). If \eqref{eq:PDDA} is designed such that its excess of passivity is greater than the lack of passivity of the economic ports obtained in Theorem~\ref{th:econ}, the feedback interconnection is asymptotically stable \cite[Theorem~6.2]{khalil2002nonlinear}. 
Note that a formal proof using the methodology described above is omitted due to space constraints.

\section{Simulation results} \label{sec:simulation}

In this section, the proposed methods are illustrated on different scenarios with networked microgrids. First, we study the operation of networked microgrids interconnected solely through the electric ports. Thereafter, we show that global optimality is obtained for the networked microgrids with an interconnection as proposed in Section~\ref{sec:interconnection_econ} for the economic ports. 

\subsection{Proposed scenario} \label{sec:simulation_scenario}

We consider a meshed DC microgrid with 8 nodes and 9 lines interconnected via electric lines (purple) to a second meshed DC microgrid with 4 nodes and 4 lines shown in Figure~\ref{fig:microgird_interconnected}, which was adapted from the scenario in \cite{zhao2015distributed}. In Microgrid 1, Node 4 (turquoise) is equipped with a grid-forming \gls{dgu}, which stabilizes the grid dynamics and acts here as price-forming entity. Nodes 2, 3 and 4 (red) have a grid-following \gls{dgu}, which are the price-taking feeders. The cost of the power injection is set to $f_1^1(p_1) = 1.2p_1^2$, $f_2^1(p_2) = 1.3p_2^2$ and $f_3^1(p_3) = 1.4p_3^2$. All other nodes (black) consist only of nonlinear loads~\eqref{eq:load}. In Microgrid 2, the grid-forming \gls{dgu} is located at Node 3. Grid-following \glspl{dgu} are located at Node 1 and 4, while Nodes 2 consists only of a nonlinear load. The cost of the power injection is set for the \glspl{dgu} in microgrid 2 to $f_1^2(p_1) = 1.4p_1^2$ and $f_4^2(p_4) = 1.5p_4^2$. 

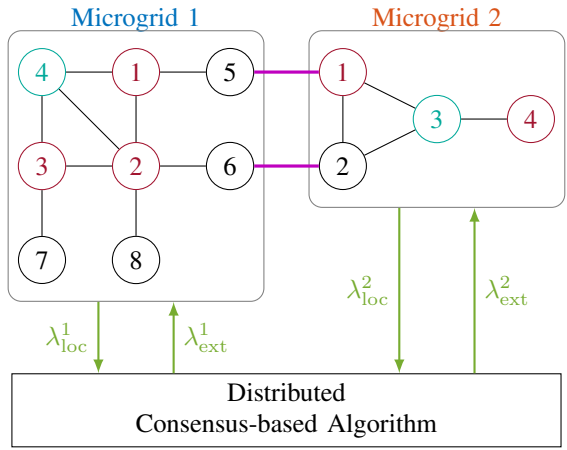
\begin{figure}[t]
    \centering
    \begin{tikzpicture}
        \node[draw, circle, gfr] (n4) at (0,0) {\textcolor{gfr}{4}};
        \node[draw, circle, gfl] (n3) at (0,-1.25) {\textcolor{gfl}{3}};
        \node[draw, circle, gfl] (n2) at (1.25,-1.25) {\textcolor{gfl}{2}};
        \node[draw, circle, gfl] (n1) at (1.25,0) {\textcolor{gfl}{1}};
        \node[draw, circle] (n7) at (0,-2.5) {7};
        \node[draw, circle] (n8) at (1.25,-2.5) {8};
        \node[draw, circle] (n5) at (2.5,0) {5};
        \node[draw, circle] (n6) at (2.5,-1.25) {6};

		\node[draw, circle, gfl] (n21) at (4,0) {\textcolor{gfl}{1}};
        \node[draw, circle] (n22) at (4,-1.25) {2};
		\node[draw, circle, gfr] (n23) at (5.25,-0.625) {\textcolor{gfr}{3}};
        \node[draw, circle, gfl] (n24) at (6.5,-0.625) {\textcolor{gfl}{4}};

        \draw[] (n1) -- (n4);
        \draw[] (n1) -- (n2);
        \draw[] (n1) -- (n5);
        \draw[] (n2) -- (n4);
        \draw[] (n2) -- (n6);
        \draw[] (n2) -- (n8);
        \draw[] (n3) -- (n7);
        \draw[] (n3) -- (n2);
		\draw[] (n3) -- (n4);

		\draw[] (n21) -- (n22);
        \draw[] (n21) -- (n23);
        \draw[] (n22) -- (n23);
        \draw[] (n23) -- (n24);

		\node[m1] (mg1) at (1.25,0.7) {Microgrid 1};
		\node[m2] at (5.25,0.7) {Microgrid 2};
		\node[draw, minimum width=3.4cm, minimum height = 3.6cm, gray, rounded corners=5] at (1.25,-1.25) {};
		\node[draw, minimum width=3.4cm, minimum height = 2.35cm, gray, rounded corners=5] at (5.25,-0.625) {};

		\draw[lines, very thick] (n21) -- (n5);
        \draw[lines, very thick] (n22) -- (n6);
        
        \node[draw, align=center, minimum width=7.3cm] (dda) at (3.25,-4.5) {Distributed \\ Consensus-based Algorithm};
	 \draw[latex-,econ, thick]  ($(dda.north)+(-1.5,0.95)$) -- node[right] {$\lambda^1_{\rm ext}$} ($(dda.north)+(-1.5,0)$) ;
	 \draw[-latex, econ, thick]  ($(dda.north)+(-2.5,0.95)$) -- node[left] {$\lambda^1_{\rm loc}$} ($(dda.north)+(-2.5,0)$) ;
	 
	 \draw[-latex, econ, thick]  ($(dda.north)+(+1.5,2.2)$) -- node[left] {$\lambda^2_{\rm loc}$} ($(dda.north)+(+1.5,0)$) ;
	 \draw[latex-, econ, thick]  ($(dda.north)+(+2.5,2.2)$) -- node[right] {$\lambda^2_{\rm ext}$} ($(dda.north)+(+2.5,0)$) ;

    \end{tikzpicture}
    \caption{Microgrids with grid-forming DGUs (turquoise), grid-following DGUs (red), and nonlinear load (black) nodes, interconnected through the electric ports with electric lines (purple) and economic ports (green)}
	\label{fig:microgird_interconnected}
\end{figure}

Load steps occur at time $t_1 = \SI{20}{\second}$ and $t_2 = \SI{40}{\second}$. The total load of the microgrids after the load steps is shown in Table~\ref{tab:load_inter}. At $t_1 = \SI{20}{\second}$, we have only an increase of the load in Microgrid 1 and hence an increase in the sum of the loads in both microgrids. At $t_2 = \SI{40}{\second}$, the load in Microgrid 1 decreases but increases in Microgrid 2, such that the sum of the loads decreases. 
\begin{table}[t]
	\centering
	\caption{Total microgrid load}
	\label{tab:load_inter}
	\begin{tabular}{c | c | c | c}
		& $\SI{0}{\second} - \SI{20}{\second}$ & $\SI{20}{\second} - \SI{40}{\second}$ & $\SI{40}{\second} - \SI{60}{\second}$ \\ \midrule
		Microgrid 1 & \SI{16}{\kilo \watt} & \SI{20.4}{\kilo \watt} & \SI{9.5}{\kilo \watt} \\
		Microgrid 2 & \SI{8}{\kilo \watt} & \SI{8}{\kilo \watt} & \SI{11}{\kilo \watt} \\ \hline
		Sum & \SI{24}{\kilo \watt} & \SI{28.4}{\kilo \watt} & \SI{20.5}{\kilo \watt}
	\end{tabular}
\end{table}
The reference voltage $v_\reff$ is set to \SI{1000}{\volt}. Typical parameter values for the lines and \gls{dgu} parameters are taken from~\cite{tucci2016}.

In the next section, two microgrids interconnected through electric ports are considered in order to highlight the plug-and-play stability of microgrids stated in Proposition~\ref{lem:electric_ports}. 

\subsection{Networked microgrids through electric ports} \label{sec:simulation_elec}

In this section, the microgrids interconnected via the electric lines (purple) as shown in Figure~\ref{fig:microgird_interconnected} are considered. For both microgrids, Proposition~\ref{lem:electric_ports} is verified for the subset of the state space defined via $v_{\min} = \SI{950}{\volt}$, $v_{\max} = \SI{1050}{\volt}$ for all nodes and $i_{\text{f},\min} = \SI{-15}{\ampere}$, $i_{\text{f},\max} = \SI{15}{\ampere}$ for all filter currents.

The power injection of the grid-following \glspl{dgu} and the price in both microgrids are shown in Figures~\ref{fig:elec_power}~and~\ref{fig:elec_preis}, respectively. When the load step at $t_1 = \SI{20}{\second}$ occurs, the \glspl{dgu} in Microgrid 1 increase their power injection, while the \glspl{dgu} in Microgrid 2 decrease their power injection, even if the sum of the loads increased (see Table~\ref{tab:load_inter}). At $t_2 = \SI{40}{\second}$, the power injection in Microgrid 1 decreases and increases in Microgrid 2, following the trend of their own load (see Table~\ref{tab:load_inter}). 

Similar dynamics are also shown by the price in Figure~\ref{fig:elec_preis}. In particular, the price dynamics in both microgrids depend on the local microgrid load changes rather than on the sum of the loads across all microgrids. The microgrids are thus not cooperating, since the economic ports are not considered. As can be observed, the prices in both microgrids are not equal at any steady-state, which entails suboptimal economic dispatch. Furthermore, the steady-state prices cannot be directly influenced, since it results from the grid-forming controllers in each microgrid, which may not have identical parameters.

In order to influence the prices in each microgrid and obtain global economic optimal dispatch, an interconnection via the economic ports as described in Section~\ref{sec:interconnection_econ} is employed in the next subsection. 

\begin{figure}[t]
	\centering
	\includegraphics{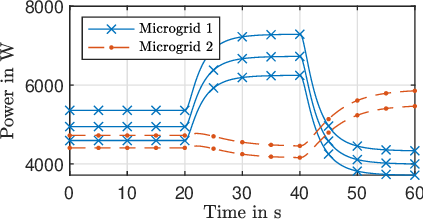}
	\caption{DGU injected power in both microgrids}
	\label{fig:elec_power}
\end{figure}

\begin{figure}[t]
	\centering
	\includegraphics{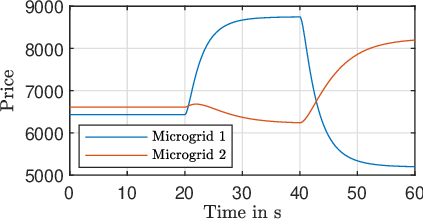}
	\caption{Local price $\lambda_\loc$ of both microgrids}
	\label{fig:elec_preis}
\end{figure}

\subsection{Economically interconnected microgrids}

In this subsection, the networked microgrids are equipped with the economic ports and the dynamic averaging as in Section~\ref{sec:interconnection_econ}. Theorem~\ref{th:econ} is verified for the same subset of the state space as in Section~\ref{sec:simulation_elec}.   

The power injection of the grid-forming \glspl{dgu} and the price in both microgrids are shown in Figures~\ref{fig:econ_power}~and~\ref{fig:econ_preis}, respectively. When a load step occurs, the prices in both microgrids vary. However, after a transient period, the prices in both microgrids are equal. This is due to the consensus algorithm~\eqref{eq:PDDA} interconnecting the microgrids via the economic port, which achieves a price consensus at steady state. The power injections of the grid-following \glspl{dgu} in Figure~\ref{fig:econ_power} adjusts automatically according to the price in the microgrids. Note that the \glspl{dgu} of Microgrid 2 inject less power than the \glspl{dgu} in Microgrid 1, since they have greater injected power costs (see Section~\ref{sec:simulation_scenario}). \glspl{dgu} with the same cost function (e.g. $f_3^1$ and $f_4^2$) inject the same amount of power in steady state despite being in different microgrids with different loads. Since at steady state a price consensus of both microgrids is achieved, we have optimal dispatch at steady state according to Proposition~\ref{lem:optimality}. Note that optimal dispatch is achieved solely by an exchange of the local prices in a distributed manner via the economic ports.

\begin{figure}[t]
	\centering
	\includegraphics[]{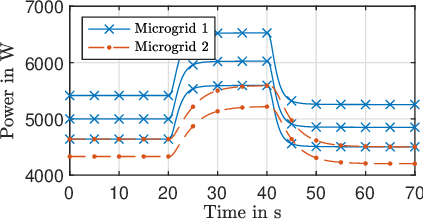}
	\caption{DGU inj. power in both microgrids interconnected economically}
	\label{fig:econ_power}
\end{figure}

\begin{figure}[t]
	\centering
	\includegraphics{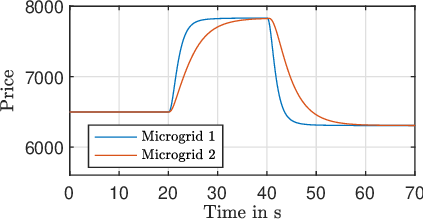}
	\caption{Local price $\lambda_\loc$ of both microgrids interconnected economically}
	\label{fig:econ_preis}
\end{figure}

\section{Conclusion}

In this paper, we proposed a price-based controller for a single microgrid achieving optimal economic dispatch. Furthermore, we proposed a novel method for studying stability of networked microgrids with an optimization-based controller. This method is based on the herein introduced electric and economic interconnection ports, whose passivity properties provide valuable insights for achieving modular stability guarantees. With the economic interconnection scheme proposed, globally optimal dispatch is achieved via the distributed communication of local prices.
Future work includes the consideration of constraints in the price-based controllers and the extension of the proposed methodology for AC microgrids.

\printbibliography

\end{document}